%% file: main.tex
\documentclass[cameraready]{Interspeech}
\title{Rethinking Entropy Minimization in Test-Time Adaptation for Autoregressive Models}
\author[equalcontribution]{Wei-Ping}{Huang}
\author[equalcontribution]{Chee-En}{Yu}
\author[correspondingauthor]{Guan-Ting}{Lin}
\author[correspondingauthor]{Hung-yi}{Lee}
\address{
    Graduate Institute of Communication Engineering, National Taiwan University, Taipei, Taiwan
}

\email{\{thomas1232121, allenyu172, daniel094144\}@gmail.com, hungyilee@ntu.edu.tw}

\keywords{Test-time Adaptation, Autoregressive Model, Automatic Speech Recognition}

\usepackage{comment}
\usepackage{mathtools}
\usepackage{amsmath, amsthm, amsfonts}
\usepackage{graphicx}
\usepackage{multirow}
\usepackage{pifont}
\usepackage[dvipsnames]{xcolor}
\usepackage{booktabs, adjustbox}
\newtheorem{theorem}{Theorem}

\begin{document}

\maketitle

\begin{abstract}
    Test-Time Adaptation (TTA) via entropy minimization (EM) has proven effective for classification tasks, yet its application to generative autoregressive models remains theoretically fragmented. Existing approaches typically rely on distinct heuristics, such as teacher forcing with pseudo labels or policy-gradient-based reinforcement learning, without a unified mathematical foundation. In this work, we resolve this discrepancy by deriving a rigorous formulation of EM tailored to autoregressive models. We show that the exact objective naturally decomposes into a token-level policy gradient loss and a token-level entropy loss, and we reinterpret prior methods as partial realizations of this unified formulation. Using Whisper ASR as a testbed, we demonstrate that our approach consistently improves performance across more than 20 diverse domains, including acoustic noise, accents, and multilingual settings.
\end{abstract}

\section{Introduction}

% TTA overview, EM definition
Test-Time Adaptation (TTA) has emerged as a promising paradigm for addressing the distribution shifts of real-world data at inference time. By adapting the source model using only test data, TTA enhances robustness without requiring access to the original training distribution. Among the various TTA types, episodic TTA~\cite{memo} focuses on the \textit{one-sample} learning problem. Specifically, for each test input $q$, the model parameters $\theta$ are optimized through an unsupervised objective $\mathcal{L}(\theta; q)$ before inferring the final prediction, after which the parameters are reset. Although various objectives have been proposed to guide this online refinement, the literature predominantly centers on entropy minimization (EM), where the model is encouraged to reduce the uncertainty of its predictive distribution $\pi_\theta$ over the input $q$:
\begin{equation}
    \mathcal{L}(\theta; q) = \mathbb{E}_{y \sim\pi_\theta(\cdot|q)} [-\log \pi_\theta(y|q)].
\end{equation}
Owing to its simplicity and effectiveness in the absence of external supervision, EM has become a fundamental building block for TTA across computer vision~\cite{memo, tent, eata, cotta}, natural language processing~\cite{agarwal2025unreasonable, gao2025oneshotentropyminimization, slot}, and speech processing~\cite{suta, sgem}.

% Problem of EM for AR models
While EM is straightforward for classification, its application to generative, autoregressive architectures is non-trivial. The primary challenge lies in the high-dimensional combinatorial nature of sequence generation, where the prediction $y$ is a sequence of dependent tokens rather than a single categorical label. This complexity has led to a fragmented landscape of implementation strategies, revealing a theoretical divide in \textbf{\textit{how the EM objective should be implemented for autoregressive models}}. In \cite{gao2025oneshotentropyminimization, sgem, cea, slmtta}, EM for autoregressive models is reduced to teacher-forcing-style adaptation using pseudo-labels. This approach takes the concatenation of the test input and pseudo-label as the full input sequence and minimizes the entropy of the output distribution specifically at each token position within the pseudo-label sequence. Conversely, in ~\cite{agarwal2025unreasonable}, EM is approached from a Reinforcement Learning (RL) perspective, treating entropy as a cost to be optimized via policy gradients. This discrepancy has created a landscape where both approaches rely on distinct heuristics while overlooking the mathematically correct gradient expression for EM. Consequently, the relationship between these heuristic methods and a theoretically rigorous EM remains unclear. We find that these approaches lead to incorrect EM that fail to capture the full probabilistic structure of autoregressive models.

\begin{figure}
  \centering
  \includegraphics[width=\linewidth]{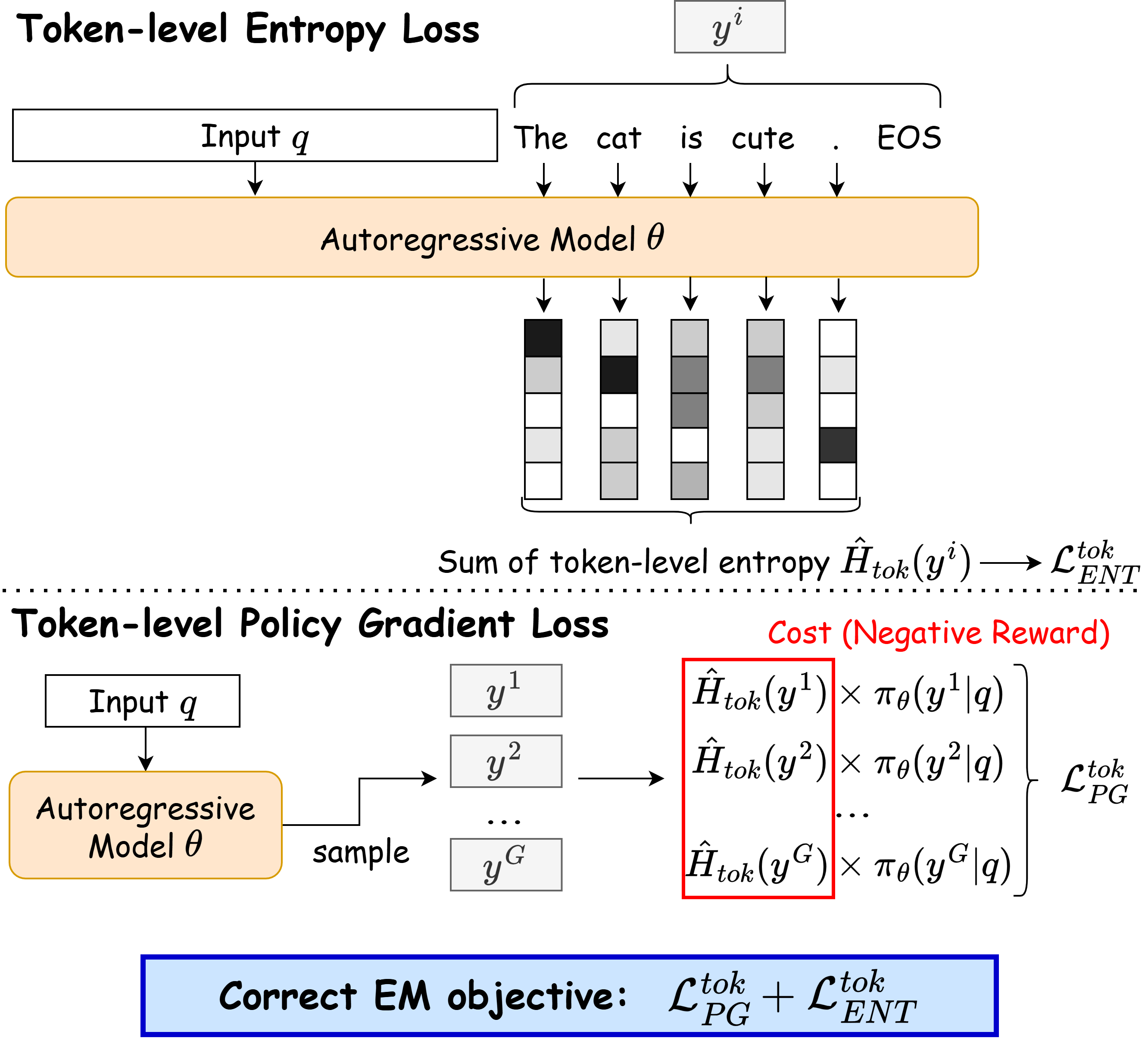}
  \caption{The correct EM objective for the autoregressive model decomposes into a token-level policy gradient loss and a token-level entropy loss. See Section~\ref{sec:method} for the details.}
  \label{fig:overall}
\end{figure}

% Our theoretical solution
In this work, we resolve the discrepancy by establishing a mathematically rigorous gradient expression of EM specifically designed for autoregressive models. By deriving the correct objective, we demonstrate that it naturally decomposes into two essential components: a token-level policy gradient loss and a token-level entropy loss. The decomposition is shown in Figure~\ref{fig:overall}. We reveal that prior teacher-forcing and RL-based approaches are not competing methodologies, but rather partial realizations of a single, unified objective. Our framework thus provides a principled theoretical foundation for EM for any autoregressive model, and clarifies the mathematical linkages to the prior works.

% Empirical validation of our solution
Though the formulation is general, we focus our empirical validation on Automatic Speech Recognition (ASR) using Whisper~\cite{whisper}, a widely adopted foundation model that serves as an ideal testbed due to its sequence-to-sequence structure. By applying the derived objective, we demonstrate that this principled approach to TTA significantly enhances performance across a wide range of scenarios. To the best of our knowledge, this is the first study to perform TTA on the Whisper model, providing a large-scale empirical validation of our theoretical framework.

Our contributions can be summarized as follows:
\begin{itemize}
    \item We identify the mathematically correct EM formulation for autoregressive models and elucidate the relationship between our proposed objective and prior heuristic-based works.

    \item Based on our formulation, we successfully perform TTA on Whisper ASR across more than 20 diverse domains, covering acoustic noise, accents, and multilingual scenarios.

    \item We provide a comprehensive comparison of different objectives, offering insights into their characteristics.
\end{itemize}

\section{Related Work}
Entropy minimization has established itself as a fundamental building block for TTA across diverse modalities. The foundational success of EM-based TTA is largely demonstrated in tasks where the model's predictive distribution is treated as a set of independent categorical variables. In computer vision, TENT~\cite{tent} pioneers optimizing batch-normalization parameters via EM, and subsequent works such as CoTTA~\cite{cotta} and SAR~\cite{sar} extend this paradigm to address challenges in continual distribution shifts. A similar trend exists in the speech community with SUTA~\cite{suta} and subsequent works~\cite{litta, suta-lm}, which introduce TTA for ASR by applying EM to non-autoregressive models like wav2vec 2.0-CTC~\cite{NEURIPS2020_92d1e1eb}. In these works, the entropy objective is a summation across independent frames or images, making the corresponding gradient computation straightforward.

As models have transitioned toward generative, autoregressive architectures where tokens exhibit temporal dependencies, implementing EM has led to different heuristics, primarily divided between teacher-forcing and RL paradigms. The first paradigm utilizes pseudo labels, taking the concatenation of the test input and pseudo-label as the full input sequence and minimizing the entropy of the output distribution specifically at each token position within the pseudo-label sequence. In natural language processing, \cite{gao2025oneshotentropyminimization} employs this one-shot EM technique for mathematical reasoning in LLMs. In speech processing, SGEM~\cite{sgem} and CEA~\cite{cea} adopt this approach to extend TTA to Transducers~\cite{burchi2021efficientconformerprogressivedownsampling}, an autoregressive ASR architecture. Recently, SLM-TTA~\cite{slmtta} also applies similar heuristics to modern generative Speech Language Models for speech translation and spoken QA. In parallel, the second paradigm approaches EM from an RL perspective, treating entropy as a scalar reward or cost to be optimized via policy gradients. Recent studies in natural language processing~\cite{empo, seedgrpo} utilize semantic entropy~\cite{kuhn2023semantic} within RL-based frameworks to guide model behavior. Notably, \cite{agarwal2025unreasonable} provides the first systematic comparison between these two heuristics under a post-training setting, highlighting the empirical differences. 

The divide between these two strategies in prior work creates ambiguity and leaves the theoretical connection between the heuristics and the true entropy minimization objective unclear. We resolve this ambiguity by providing a rigorous derivation that clarifies their relationship.

% \subsection{On the Pitfalls of Gradient Estimation}
% The mathematical treatment of sequence-level objectives often centers on the challenge of differentiating expectations over discrete combinatorial spaces. Our analysis builds upon a similar line of reasoning in recent work examining the pitfalls of KL-divergence gradient estimation for autoregressive models~\cite{tang2025pitfallskldivergencegradient}. Just as \cite{tang2025pitfallskldivergencegradient} showed that common implementations frequently omit critical terms in the gradient expression, we demonstrate that a parallel issue exists for the EM objective in TTA for autoregressive models. Our work thus provides a theoretical foundation for TTA that has been lacking in the transition from simple classification settings to sequence-level generation.

\section{Method}
\label{sec:method}
We first establish the formal notation for autoregressive generation and the entropy estimation in Section~\ref{sec:math-est}. In Section~\ref{sec:math-em-tok} and \ref{sec:math-em-seq}, we present our primary theoretical contribution of the derivation of the mathematically complete gradient expression of EM for autoregressive models. Section~\ref{sec:link} elucidates the theoretical linkage between our unified formulation and prior teacher-forcing and RL-based approaches. Finally, we incorporate practical implementation tricks in Section~\ref{sec:tricks}.

\subsection{Estimate Entropy for Autoregressive Models}
\label{sec:math-est}
Denote $\Omega$ as all possible output sequences with finite length that end with the EOS token. Denote $\pi_{\theta}(\cdot|q)$ as the autoregressive model's output distribution over $\Omega$ when the input is $q$~\footnote{In ASR, input $q$ is the speech and output $\bm{y}$ is the transcription.}. For any sequence $\bm{y}=y_1y_2\cdots y_{|\bm{y}|}\in\Omega$, denote its probability as $\pi_{\theta}(\bm{y}|q)$. Note that we drop the notation $q$ for convenience since it is fixed throughout the analysis. The entropy of the model with respect to input $q$ can be written as
\begin{align}
    H(\pi_{\theta})=\mathbb{E}_{\bm{y}\sim \pi_{\theta}}[-\log (\pi_{\theta}(\bm{y}))].
\end{align}
Following~\cite{agarwal2025unreasonable}, we consider two methods for empirically estimating the entropy.

\noindent\textbf{Sequence-level entropy estimator} is computed based on the distribution of the output sequences produced by the model. In practice, it samples a full sequence $\bm{y}\sim\pi_{\theta}$ and computes the total log-probability.
\begin{align}
    \hat{H}_{seq}(\bm{y}) = -\log\pi_{\theta}(\bm{y}) = -\sum_{t=1}^{|\bm{y}|}\log\pi_{\theta}(y_t|\bm{y}_{<t}),
\end{align}
where $\bm{y}$ is sampled from $\pi_{\theta}$.

\noindent\textbf{Token-level entropy estimator} computes the entropy based on the token distribution at each step. In practice, it sums per-token entropy across time steps. At generation step $t$, the model produces a probability distribution over a vocabulary $\mathcal{V}$. The estimator is defined as
\begin{align}
    \hat{H}_{tok}(\bm{y}) = \sum_{t=1}^{|\bm{y}|}\mathcal{H}(\pi_{\theta}(\cdot|\bm{y}_{<t})).
\end{align}
where $\mathcal{H}(\pi_{\theta}(\cdot|\bm{y}_{<t}))= -\sum_{j\in\mathcal{V}}\pi_{\theta}(j|\bm{y}_{<t})\log\pi_{\theta}(j|\bm{y}_{<t})$.

While these estimators appear frequently in existing literature~\cite {gao2025oneshotentropyminimization, sgem, cea, slmtta}, a rigorous justification for their validity is often omitted. Specifically, the unbiasedness of the sequence-level estimator, $\hat{H}_{seq}$, follows directly from the definition of the entropy. However, the unbiasedness of the token-level estimator, $\hat{H}_{tok}$, is less trivial. To ground our method in theoretical consistency, we provide a formal proof below.

\begin{theorem}
\label{thm:unbias}
If the entropy is finite, i.e., $H(\pi_\theta) < \infty$, then the token-level entropy estimator $\hat{H}_{tok}$ is an unbiased estimator of the entropy $H(\pi_\theta)$.
\end{theorem}
\begin{proof}
Since $H(\pi_\theta) < \infty$, by the monotone convergence theorem, we can decompose sequence-level entropy into an infinite sum of token-wise terms.
\begin{subequations}
    \begin{align}
        H(\pi_{\theta}) &= \mathbb{E}_{\bm{y}\sim \pi_{\theta}} \left[\sum_{t=1}^{|\bm{y}|}-\log\pi_{\theta}(y_t|\bm{y}_{<t})\right]\\
        &=\sum^{\infty}_{t=1}\mathbb{E}_{\bm{y}\sim \pi_{\theta}}[-\mathbb{I}[|\bm{y}|\geq t]\log\pi_{\theta}(y_t|\bm{y}_{<t})], \label{eq:term}
    \end{align}
\end{subequations}
where $\mathbb{I}$ is the indicator function.

Let $\bm{Y}$ be the random variable representing the outcome sampled from $\pi_{\theta}$, $\bm{Y}_{<k+1}$ be the random variable representing the prefix of $\bm{Y}$ with length $k$ (if $\bm{Y}$ has length less than $k$, assign $\bm{Y}_{<k+1}=\phi$), and $\bm{Y}_k$ be the random variable representing the $k$-th token of $\bm{Y}$. By the law of total expectation, condition the prefix $\bm{Y}_{<t}$ on the $t$-th token-wise term and get
\begin{subequations}
    \begin{align}
        &\mathbb{E}_{\bm{y}\sim \pi_{\theta}}[-\mathbb{I}[|\bm{y}|\geq t]\log\pi_{\theta}(y_t|\bm{y}_{<t})] \label{eq:est-a} \\
        &=\mathop{\mathbb{E}}_{\bm{Y}_{<t}}\left[\mathop{\mathbb{E}}_{\bm{Y}|\bm{Y}_{<t}}\bigg[-\mathbb{I}[|\bm{Y}|\geq t]\log\pi_{\theta}(\bm{Y}_t|\bm{Y}_{<t})\bigg]\right] \label{eq:est-b} \\
        &=\mathop{\mathbb{E}}_{\bm{Y}_{<t}}\left[\mathbb{I}[|\bm{Y}|\geq t]\mathcal{H}(\pi_{\theta}(\cdot|\bm{Y}_{<t}))\right] \label{eq:est-c} \\
        &=\mathop{\mathbb{E}}_{\bm{Y}_{<t}}\left[\mathop{\mathbb{E}}_{\bm{Y}|\bm{Y}_{<t}}\bigg[\mathbb{I}[|\bm{Y}|\geq t]\mathcal{H}(\pi_{\theta}(\cdot|\bm{Y}_{<t}))\bigg]\right] \label{eq:est-d} \\
        &=\mathbb{E}_{\bm{y}\sim\pi_{\theta}}\left[\mathbb{I}[|\bm{y}|\geq t]\mathcal{H}(\pi_{\theta}(\cdot|\bm{y}_{<t}))\right], \label{eq:est-e}
    \end{align}
\end{subequations}

To see that \ref{eq:est-b} equals \ref{eq:est-c}, consider two cases. When $\bm{Y}_{<t}=\phi$ or $\bm{Y}_{<t}$ ends with the EOS token, $\mathbb{I}[|\bm{Y}|\geq t]=0$; otherwise $\mathbb{I}[|\bm{Y}|\geq t]=1$, and the conditional expectation given prefix $\bm{Y}_{<t}$ is equivalent to the conditional expectation over the next token, which is exactly the token entropy at the $t$-th step. To see that \ref{eq:est-c} equals \ref{eq:est-d}, note that $\mathbb{I}[|\bm{Y}|\geq t]$ is fully determined by prefix $\bm{Y}_{<t}$, so adding the conditional expectation remains unchanged. \ref{eq:est-d} equals \ref{eq:est-e} again follows from the law of total expectation.

Finally, replace each term in \ref{eq:term} with \ref{eq:est-e} and swap the infinite sum and the expectation back by the monotone convergence theorem.
\begin{subequations}
    \begin{align}
        H(\pi_{\theta})&=\sum^{\infty}_{t=1}\mathbb{E}_{\bm{y}\sim \pi_{\theta}}[-\mathbb{I}[|\bm{y}|\geq t]\log\pi_{\theta}(y_t|\bm{y}_{<t})]\\
        &=\sum_{t=1}^{\infty}\mathbb{E}_{\bm{y}\sim\pi_{\theta}}[\mathbb{I}[|\bm{y}|\geq t]\mathcal{H}(\pi_{\theta}(\cdot|\bm{y}_{<t}))]\\
        &=\mathbb{E}_{\bm{y}\sim\pi_{\theta}}\left[\sum_{t=1}^{|\bm{y}|}\mathcal{H}(\pi_{\theta}(\cdot|\bm{y}_{<t}))\right]\\
        &=\mathbb{E}_{\bm{y}\sim\pi_{\theta}}[\hat{H}_{tok}(\bm{y})].
    \end{align}
\end{subequations}

The token-level entropy estimator $\hat{H}_{tok}$ is an unbiased estimator of the entropy $H(\pi_\theta)$, as we desired.
\end{proof}

\subsection{Minimizing Entropy by Differentiating through the Token-level Estimator}
\label{sec:math-em-tok}
To perform entropy minimization, take the derivative through the expectation over the token-level entropy estimator.
% \begin{equation}
%     \resizebox{1.0\linewidth}{!}{$
%         \begin{aligned}
%             &\nabla_{\theta}H(\pi_{\theta})\\
%             &=\nabla_{\theta} \mathbb{E}_{x\sim \pi_{\theta}}[\hat{H}_{tok}(x)] \\
%             &= \int_x \nabla_{\theta}(\pi_{\theta}(x)\hat{H}_{tok}(x))dx \\
%             &= \int_x \hat{H}_{tok}(x)\nabla_{\theta}\pi_{\theta}(x) dx + \int_x\pi_{\theta}(x)\nabla_{\theta}\hat{H}_{tok}(x)dx \\
%             &= \int_x \pi_{\theta}(x)\hat{H}_{tok}(x)\nabla_{\theta}\log\pi_{\theta}(x) dx + \int_x \pi_{\theta}(x)\nabla_{\theta}\hat{H}_{tok}(x)dx \\
%             &= \mathbb{E}_{x\sim \pi_{\theta}}[\hat{H}_{tok}(x)\nabla_{\theta}\log\pi_{\theta}(x)] + \mathbb{E}_{x\sim \pi_{\theta}}[\nabla_{\theta}\hat{H}_{tok}(x)].
%         \end{aligned}$
%     }    
% \end{equation}
\begin{subequations}
    \begin{align}
        &\nabla_{\theta}H(\pi_{\theta})\\
        &=\nabla_{\theta} \mathbb{E}_{\bm{y}\sim \pi_{\theta}}[\hat{H}_{tok}(\bm{y})] \\
        &= \int_{\bm{y}} \nabla_{\theta}(\pi_{\theta}(\bm{y})\hat{H}_{tok}(\bm{y}))d\bm{y} \label{eq:grad-c} \\
        &= \int_{\bm{y}} \hat{H}_{tok}(\bm{y})\nabla_{\theta}\pi_{\theta}(\bm{y}) d\bm{y} + \int_{\bm{y}}\pi_{\theta}(\bm{y})\nabla_{\theta}\hat{H}_{tok}(\bm{y})d\bm{y} \label{eq:grad-d} \\
        \begin{split}
            &= \int_{\bm{y}} \pi_{\theta}(\bm{y})\hat{H}_{tok}(\bm{y})\nabla_{\theta}\log\pi_{\theta}(\bm{y}) d\bm{y} \\
            &\quad + \int_{\bm{y}} \pi_{\theta}(\bm{y})\nabla_{\theta}\hat{H}_{tok}(\bm{y})d\bm{y} 
        \end{split} \label{eq:grad-e} \\
        \begin{split}
        &= \mathbb{E}_{\bm{y}\sim \pi_{\theta}}[\hat{H}_{tok}(\bm{y})\nabla_{\theta}\log\pi_{\theta}(\bm{y})] \\
        &\quad + \mathbb{E}_{\bm{y}\sim \pi_{\theta}}[\nabla_{\theta}\hat{H}_{tok}(\bm{y})].
        \end{split}
    \end{align}
\end{subequations}
\ref{eq:grad-c} to \ref{eq:grad-d} follows from the chain rule, and \ref{eq:grad-d} to \ref{eq:grad-e} is the log-derivative trick.
The corresponding loss function is
\begin{align}
    \mathcal{L}_{EM}^{tok} = \underbrace{sg(\hat{H}_{tok}(\bm{y})) \log \pi_{\theta}(\bm{y})}_{\mathcal{L}_{PG}^{tok}} + \underbrace{\hat{H}_{tok}(\bm{y}) }_{\mathcal{L}_{ENT}^{tok}},
    \label{eq:em-tok}
\end{align}
where $sg(\cdot)$ is the stop gradient operator.

The loss function $\mathcal{L}_{EM}^{tok}$ decomposes into two components that address different aspects of the sequence distribution. We denote the first component as the \textbf{token-level policy gradient loss, $\bm{\mathcal{L}_{PG}^{tok}}$}. This term arises from the derivative of the expectation with respect to the sampling distribution, which is identical to the loss used in REINFORCE~\cite{REINFORCE} when the token-level entropy estimator is treated as cost. Consequently, $\mathcal{L}_{PG}^{tok}$ serves as a signal that adjusts the probability of generated trajectories to favor those with lower uncertainty.

The second component is the \textbf{token-level entropy loss, $\bm{\mathcal{L}_{ENT}^{tok}}$}, which represents the direct minimization of the estimator itself via the pathwise gradient. In stochastic optimization, the pathwise gradient accounts for the dependence of the cost function on the parameters, assuming the samples are fixed. It is important to note that while $\mathcal{L}_{ENT}^{tok}$ appears as a standard unsupervised loss, directly optimizing this term is not equivalent to minimizing the full expected entropy since the tokens are sampled from a distribution that is itself parameterized by $\theta$. Therefore, the token-level policy gradient loss $\mathcal{L}_{PG}^{tok}$ is required to account for how changes in the model parameters shift the distribution of trajectories.

\subsection{Minimizing Entropy by Differentiating through the Sequence-level Estimator}
\label{sec:math-em-seq}
On the other hand, we can also perform entropy minimization by taking the derivative through the expectation over the sequence-level entropy estimator.
\begin{subequations}
    \begin{align}
        &\nabla_{\theta}H(\pi_{\theta})\\
        &= \nabla_{\theta} \mathbb{E}_{\bm{y}\sim \pi_{\theta}}[\hat{H}_{seq}(\bm{y})] \\
        &= \int_{\bm{y}} \nabla_{\theta}(\pi_{\theta}(\bm{y})\hat{H}_{seq}(\bm{y}))d\bm{y} \\
        &= \int_{\bm{y}} \hat{H}_{seq}(\bm{y})\nabla_{\theta}\pi_{\theta}(\bm{y}) d\bm{y} + \int_{\bm{y}}\pi_{\theta}(\bm{y})\nabla_{\theta}\hat{H}_{seq}(\bm{y})d\bm{y} \\
        \begin{split}
            &= \int_{\bm{y}} \pi_{\theta}(\bm{y})\hat{H}_{seq}(\bm{y})\nabla_{\theta}\log\pi_{\theta}(\bm{y}) d\bm{y} \\
            &\quad + \int_{\bm{y}} \pi_{\theta}(\bm{y})\nabla_{\theta}\hat{H}_{seq}(\bm{y})d\bm{y}
        \end{split} \\
        \begin{split}
            &= \mathbb{E}_{\bm{y}\sim \pi_{\theta}}[\hat{H}_{seq}(\bm{y})\nabla_{\theta}\log\pi_{\theta}(\bm{y})] \\
            &\quad + \mathbb{E}_{\bm{y}\sim \pi_{\theta}}[\nabla_{\theta}\hat{H}_{seq}(\bm{y})].
        \end{split}
    \end{align}
\end{subequations}
Note that the second term vanishes, since
\begin{equation}
    \begin{aligned}
        \mathbb{E}_{\bm{y}\sim \pi_{\theta}}[\nabla_{\theta}\hat{H}_{seq}(\bm{y})] &= -\int_{\bm{y}} \pi_{\theta}(\bm{y})\nabla_{\theta}\log\pi_{\theta}(\bm{y})d\bm{y} \\
        &= -\int_{\bm{y}} \nabla_{\theta}\pi_{\theta}(\bm{y})d\bm{y} \\
        &= -\nabla_{\theta}\int_{\bm{y}}\pi_{\theta}(\bm{y})d\bm{y} \\
        &= -\nabla_{\theta}(1) = 0
    \end{aligned}
\label{eq:mc-vanish}
\end{equation}
The corresponding loss function becomes
\begin{align}
    \mathcal{L}_{EM}^{seq} = \underbrace{sg(\hat{H}_{seq}(\bm{y}))\log\pi_{\theta}(\bm{y})}_{\mathcal{L}_{PG}^{seq}}.
    \label{eq:em-seq}
\end{align}

Similarly, the only term $\mathcal{L}_{PG}^{seq}$ serves as the reinforcement signal, which guides the model to favor trajectories with higher probability. We denote $\mathcal{L}_{PG}^{seq}$ as the \textbf{sequence-level policy gradient loss}. While $\mathcal{L}_{EM}^{seq}$ and the previously defined token-level objective $\mathcal{L}_{EM}^{tok}$ appear as distinct formulations, they are functionally equivalent under gradient-based optimization. This stems from the fact that their respective gradients with respect to $\theta$ are both unbiased estimators for EM.

\subsection{Linkage to Existing Works: A Unified Perspective} 
\label{sec:link}
The derivation in Section~\ref{sec:math-em-tok} and \ref{sec:math-em-seq} allows us to position existing EM heuristics within a unified framework. In methods such as \cite{gao2025oneshotentropyminimization, sgem, cea, slmtta}, the adaptation process utilizes a teacher-forcing heuristic where the model first generates a pseudo-label by greedy decoding, and directly adopts the token-level entropy estimator as the loss function. Within our framework, this corresponds to optimizing only the token-level entropy loss $\mathcal{L}_{ENT}^{tok}$. As discussed earlier, optimizing $\mathcal{L}_{ENT}^{tok}$ in isolation is not mathematically equivalent to minimizing the full expected entropy, as it does not account for how parameter updates alter the probability of the sampled trajectories themselves.

Conversely, the RL-based heuristic presented in \cite{agarwal2025unreasonable} focuses only on the stochastic components of the gradient. Specifically, their EM-RL-token approach treats the token-level entropy estimator as a cost signal, which corresponds precisely to our token-level policy gradient loss $\mathcal{L}_{PG}^{tok}$.

Interestingly, \cite{agarwal2025unreasonable} also evaluates a sequence-level variant (EM-RL-sequence) and a fine-tuning variant (EM-FT). The EM-FT variant is essentially equivalent to the teacher-forcing heuristic, whereas EM-RL-sequence corresponds to our sequence-level policy gradient loss $\mathcal{L}_{PG}^{seq}$. In the sequence-level variant, the RL-based approach is inherently correct because the second term in the sequence-level derivation vanishes (see Eq.~\ref{eq:mc-vanish}), leaving the policy gradient as the sole component required for the exact gradient estimator.

Our derivation shows that EM for autoregressive models requires jointly optimizing both $\mathcal{L}_{PG}^{tok}$ and $\mathcal{L}_{ENT}^{tok}$. By relating our formulation to prior work, as summarized in Table~\ref{tab:link}, we observe that existing methods often optimize only a partial version of this objective.
\input{tables/link}

\subsection{Implementation Tricks}
\label{sec:tricks}
\subsubsection{Apply RL Baseline}
\label{sec:tricks-rl}
The REINFORCE estimator is notoriously difficult to optimize due to its high variance; consequently, incorporating a baseline is a standard technique for variance reduction. Following~\cite{agarwal2025unreasonable}, we utilize the leave-one-out baseline~\cite{rloo} and apply the token-level normalization strategy from DAPO~\cite{dapo} to the token-level policy gradient loss. For each input $q$, we sample $G$ responses $\bm{y}^1, \bm{y}^2, \cdots, \bm{y}^G$. The final token-level objective becomes
\begin{align}
    \mathcal{L}_{EM}^{tok} = \frac{1}{\sum_{i=1}^G |\bm{y}^i|} \sum_{i=1}^{G} \left(A_{tok}(\bm{y}^i) \log\pi_{\theta}(\bm{y}^i) + \hat{H}_{tok}(\bm{y}^i)\right),
\end{align}
where advantage
\begin{align}
    A_{tok}(\bm{y}^i) = sg(\hat{H}_{tok}(\bm{y}^i)) - \frac{1}{G-1}\sum_{j\neq i} sg(\hat{H}_{tok}(\bm{y}^j)).
\end{align}
Similarly, for the sequence-level objective, we have
\begin{align}
    \mathcal{L}_{EM}^{seq} = \frac{1}{\sum_{i=1}^G |\bm{y}^i|}\sum_{i=1}^{G} A_{seq}(\bm{y}^i) \log\pi_{\theta}(\bm{y}^i),
\end{align}
where advantage
\begin{align}
    A_{seq}(\bm{y}^i) = sg(\hat{H}_{seq}(\bm{y}^i)) - \frac{1}{G-1}\sum_{j\neq i} sg(\hat{H}_{seq}(\bm{y}^j)).
\end{align}

\subsubsection{Extension with Beam Search}
\label{sec:tricks-beam}
While the unbiasedness of $\hat{H}_{tok}$ is established under the assumption of random sampling responses from the output distribution $\pi_{\theta}$, we introduce a practical extension by running the algorithm over responses derived from beam search instead of random sampling, that is, take the top $G$ beams as $\bm{y}^1, \bm{y}^2, \cdots, \bm{y}^G$. This approach is conceptually analogous to priority sweeping in reinforcement learning, where the optimization prioritizes the most probable responses to facilitate faster convergence on high-quality signals. Although restricting the expectation to the top $G$ beams inevitably introduces a bias, and thus does not strictly adhere to the formulation in Section~\ref{sec:math-em-tok}, it serves as an effective trick in practice as shown in Section~\ref{sec:results}.

\section{Experiments}
\subsection{Datasets}
We focus our empirical validation on the ASR task. We experiment on three datasets:
%Corrupted Librispeech, L2Arctic, and Multilingual LibriSpeech. They correspond to different noises, accents, and languages, respectively.

\noindent\textbf{Corrupted Librispeech (LS-C)~\cite{dsuta}:} The dataset is constructed by adding noises from MS-SNSD into Librispeech \texttt{test-other} set. The noises include air conditioner (AC), airport announcement (AA), babble (BA), copy machine (CM), munching (MU), neighbors (NB), shutting door (SD), typing (TP), vacuum cleaner (VC), and Gaussian noise (GS), resulting in 10 different noises in total. The Signal-to-Noise Ratio (SNR) is set to 10 dB as in \cite{dsuta}.

\noindent\textbf{L2Arctic~\cite{l2arctic}:} A non-native English speech corpus consisting of utterances from second language (L2) learners originating from 6 countries with different first languages (L1): Arabic, Mandarin, Hindi, Korean, Spanish, and Vietnamese.

\noindent\textbf{Multilingual LibriSpeech (MLS)~\cite{mls}:} A multilingual speech corpus consisting of utterances in 7 languages: Dutch, French, German, Italian, Polish, Portuguese, and Spanish.

\subsection{Methods}
\label{sec:exp-methods}
To validate the derived objectives for EM, we adopt them for TTA on Whisper ASR. For each sample, we iteratively perform adaptation on Whisper for several steps before inferring the final prediction, after which the parameters are reset. During each adaptation step, we first generate $G$ candidate transcriptions. Next, we estimate the entropy using the proposed entropy estimator and perform EM using the objectives derived in Section~\ref{sec:math-em-tok} and \ref{sec:math-em-seq}. After updating the model parameters, we resample $G$ new transcriptions using the updated model for the subsequent step. We compare the following three methods:

% To evaluate the impact of the sequence-level objective ($\mathcal{L}_{EM}^{seq}$) and the token-level objective ($\mathcal{L}_{EM}^{tok}$), as well as the effect of different sampling strategies on TTA, we compare the following three methods:

\noindent\textbf{EM-seq:} This method first randomly samples $G$ transcriptions and calculates the entropy using a \textbf{sequence-level entropy estimator}. Entropy minimization is then performed using $\mathcal{L}_{EM}^{seq}$ as the loss function.

\noindent\textbf{EM-tok:} This method first randomly samples $G$ transcriptions and calculates the entropy using a \textbf{token-level entropy estimator}. Entropy minimization is then performed using $\mathcal{L}_{EM}^{tok}$ as the loss function.

\noindent\textbf{EM-tok-b:} Extend EM-tok by using beam search to generate the top $G$ candidate transcriptions instead of random sampling. %The remaining adaptation procedures are identical to EM-tok.

% greedy-EM
In TTA for ASR~\cite{sgem, cea, slmtta}, the teacher-forcing heuristic is the standard approach for EM. We therefore adopt this heuristic as our primary baseline. Precisely, given speech input $q$, for each adaptation step, we first generate a single transcription $\hat{\bm{y}}$ via greedy decoding, then minimizes the token-level entropy $\mathcal{L}_{ENT}^{tok}$ based on $\hat{\bm{y}}$. We denote this baseline method as \textbf{Greedy-EM}.
Note that we do not directly compare our approach to the methods in \cite{sgem, cea, slmtta} since these methods incorporate additional components, such as confidence thresholding~\cite{slmtta}, or other loss functions~\cite{sgem, cea}.
We establish a simplified baseline that uses only EM to isolate its precise impact on TTA.

\subsection{Implementation Details}
The implementations of EM-seq, EM-tok, and EM-tok-b follow the descriptions in Section~\ref{sec:exp-methods}. All hyperparameters are chosen through a held-out subset from CommonVoice~\cite{cv}, and we apply this optimal configuration across all target corpora.

We use the Whisper-base\footnote{https://huggingface.co/openai/whisper-base} model as the default, the number of adaptation steps is set to 10, and the number of sampled sequences is $G=16$. We optimize the model using the AdamW optimizer with a learning rate of 1e-3. During the adaptation phase, gradients are backpropagated solely to update Whisper’s LayerNorm parameters, encompassing both the encoder and decoder LayerNorms. After the 10-step update, we decode the utterance using greedy decoding on the adapted model to report the final word error rate. All experiments are conducted on a single NVIDIA GeForce RTX 3090 GPU.

\subsection{Results}
\label{sec:results}

\begin{table*}[t]
    \centering
    
    % --- 第一個表格：Noise ---
    \caption{WER (\%) of different TTA methods across 10 noises of Corrupted LibriSpeech dataset.}
    \label{tab:ls-results}
    \begin{adjustbox}{max width=\textwidth}
        \input{tables/ls-big-table} 
    \end{adjustbox}
    
    \vspace{4mm}
    
    % --- 第二個表格：Accent ---
    \caption{WER (\%) of different TTA methods on L2-Arctic dataset.}
    \label{tab:l2-results}
    \begin{adjustbox}{max width=\textwidth}
        \input{tables/l2-big-table}
    \end{adjustbox}
    
    \vspace{4mm}
    
    % --- 第三個表格：Multilingual ---
    \caption{WER (\%) of different TTA methods on Multilingual LibriSpeech dataset.}
    \label{tab:mls-results}
    \begin{adjustbox}{max width=\textwidth}
        \input{tables/mls-big-table}
    \end{adjustbox}
    
\end{table*}

% Tables~\ref{tab:ls-results}--\ref{tab:mls-results} compare the non-adapted baseline, and four test-time adaptation (TTA)
% variants—Greedy-EM, EM-seq, EM-tok, and EM-tok—b under three kinds of distribution
% shift: additive noise, English accents, and cross-lingual transfer.

\noindent\textbf{Adaptation to Additive Noise.}
As shown in Table~\ref{tab:ls-results}, the proposed EM-seq and EM-tok mitigate the impact of acoustic degradation across the ten noise conditions. The unadapted source model yields an average Word Error Rate (WER) of 22.53\,\%, and Greedy-EM reduces the average WER to 21.91\,\%. EM-seq and EM-tok yield lower average WERs of 21.34\,\% and 20.77\,\%, respectively. Moreover, extending the token-level objective with beam search (EM-tok-b) achieves an average WER of 19.15\,\% and represents the lowest WER across all ten noise conditions.

\noindent\textbf{Adaptation to Accent Shifts.}
Table~\ref{tab:l2-results} presents the performance on accented speech. The unadapted source model yields an average WER of 19.35\,\%, and Greedy-EM achieves an average WER of 18.77\,\%, but exhibits variance across different accent corpora; for example, it yields a higher WER than the source model on the Spanish split (19.47\,\% $>$ 18.31\,\%). EM-seq and EM-tok outperform both baselines, reducing the average WERs to 17.68\,\% and 17.05\,\%, respectively. Moreover, extending the token-level objective with beam search (EM-tok-b) yields the lowest average WER of 16.21\,\%.

\noindent\textbf{Generalization beyond English.} 
Table~\ref{tab:mls-results} evaluates the generalization of our TTA methods across different languages. Across these languages, Greedy-EM reduces the average WER from 24.58\,\% to 24.08\,\%. The average WERs for EM-seq and EM-tok are 24.17\,\% and 24.07\,\%, respectively. All three methods show minor improvements, which may be due to the smaller domain shift. Despite that, EM-tok-b still shows strong performance, achieving an average WER of 22.63\%.

\noindent\textbf{Overall Comparison.}
The results demonstrate that our proposed methods are effective across all three corpora, encompassing different noises, accents, and languages. By utilizing the mathematically complete gradient expression, EM-seq achieves comparable or lower WER than both baselines across 23 domains. Furthermore, EM-tok outperforms both baselines on nearly all domains, with the exception of German and Portuguese in MLS. Even in challenging scenarios where Greedy-EM fails, such as Spanish accents in L2-Arctic, both EM-seq and EM-tok continue to yield improvements. All these evidences showcase that our proposed methods are superior to existing heuristics.

On the other hand, we observe that EM-tok consistently yields lower WERs compared to EM-seq across the three corpora. We hypothesize that this is because ASR transcription errors predominantly occur at the token level. Consequently, minimizing through the token-level entropy estimator may provide more granular guidance for correction by considering the entire vocabulary distribution at each time step. In contrast, minimizing through the sequence-level estimator primarily considers the likelihood of the decoded tokens, which may offer less information for error correction in the ASR task.

Finally, by utilizing beam search transcriptions, EM-tok-b achieves the lowest WER among all methods. Optimizing over these higher-quality trajectories allows EM-tok-b to achieve better optimization with larger performance gains. While this approach introduces a degree of bias into the gradient estimation, it serves as a highly effective trick empirically. We provide further discussions on the impact of beam search in Section~\ref{sec:beam-search} and analyze its efficiency in Section~\ref{sec:eff}.

\section{Discussion}
\subsection{Components in Token-level Objective}
\label{sec:objectives-token}
\begin{figure*}
  \centering
  \includegraphics[width=0.9\linewidth]{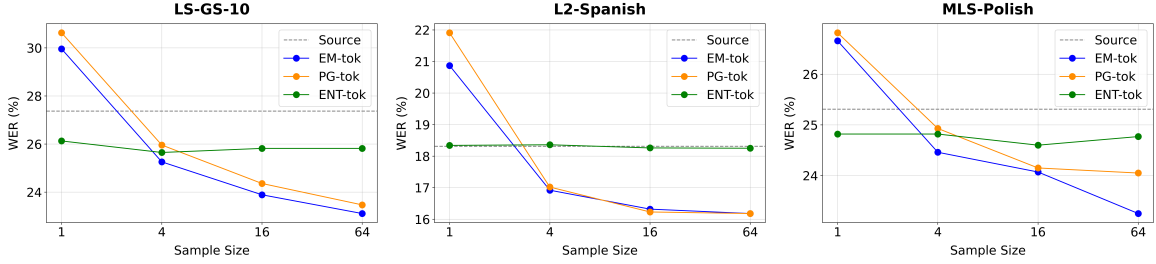}
  \caption{WER(\%) across varying sample sizes for Gaussian Noise (left), Spanish Accent (middle), and Polish (right) domains. Compares the full token-level objective (EM-tok) against its two components (PG-tok, ENT-tok).}
  \label{fig:sample-ablate0}
\end{figure*}
\begin{figure*}
  \centering
  \includegraphics[width=0.9\linewidth]{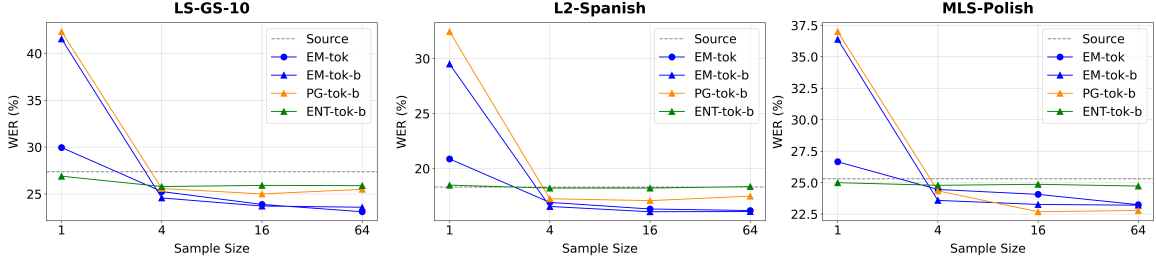}
  \caption{WER(\%) across varying sample sizes for Gaussian Noise (left), Spanish Accent (middle), and Polish (right) domains. Compares EM-tok and the variants utilizing beam search transcriptions (EM-tok-b, PG-tok-b, and ENT-tok-b).}
  \label{fig:sample-ablate1}
\end{figure*}
To investigate the interaction between the token-level policy gradient loss and the token-level entropy loss, we evaluate them across varying sample sizes ($G \in \{1, 4, 16, 64\}$) under three distinct domains: Gaussian Noise (LS-GS-10), Spanish accents (L2-Spanish), and Polish (MLS-Polish). We compare EM-tok against two ablation baselines, \textbf{PG-tok}: using only the token-level policy gradient loss, and \textbf{ENT-tok}: using only the token-level entropy loss. Note that for the case $G=1$, the baseline trick described in Section~\ref{sec:tricks-rl} is unavailable and thus omitted.

As illustrated in Figure~\ref{fig:sample-ablate0}, when multiple samples are available, the token-level policy gradient loss provides substantial improvements, with both EM-tok and PG-tok consistently outperforming ENT-tok. Furthermore, the full EM-tok objective generally yields the best results, empirically validating that the mathematically complete gradient expression is superior to the heuristics used in prior work. However, with only a single sample, the token-level policy gradient loss becomes unstable. In this specific scenario, EM-tok and PG-tok are inferior to ENT-tok. Notably, while an increased sample size significantly stabilizes EM-tok and PG-tok, ENT-tok maintains nearly constant performance across all sample sizes. These results underscore that using the correct gradient expression with sufficient sampling is essential for successful EM in autoregressive models.

\subsection{Utilizing Beam Search Transcriptions}
\label{sec:beam-search}
It is important to determine whether the performance gains of the proposed EM-tok-b arise solely from the use of beam search transcriptions during adaptation. To investigate this, we compare EM-tok-b with a beam search baseline (beam size $=16$) that performs no parameter updates. Due to space limitations, we only report partial results on two randomly selected domains from each dataset.

\input{tables/beam-search}

As shown in Table~\ref{tab:em-vs-beam}, although beam search itself is an effective method for improving ASR outputs, EM-tok-b achieves superior performance in the majority of scenarios, particularly in domains characterized by significant acoustic and accent shifts. Beam search remains highly competitive and occasionally performs better in specific languages. These results suggest that decoding strategies and adaptation methods address different aspects of the problem. Understanding why certain domains favor one approach over the other, and how they can be effectively combined, requires further investigation.

We further examine how utilizing beam search transcriptions, as proposed in Section~\ref{sec:tricks-beam}, interacts with the token-level objectives. As shown in Figure~\ref{fig:sample-ablate1}, the overall trends remain consistent with those observed in Section~\ref{sec:objectives-token}. Additionally, utilizing beam search transcriptions (EM-tok-b) provides practical advantages, generally outperforming the random sampling variant (EM-tok), particularly at smaller sample sizes ($G\in\{4,16\}$). This advantage diminishes as the sample size increases to 64, where EM-tok benefits from more accurate gradient estimation.

A notable case arises in the Polish domain, where PG-tok-b demonstrates particularly strong performance. We hypothesize that this behavior is related to the high quality of beam search decoding in this domain, as evidenced by the results in Table~\ref{tab:em-vs-beam}. This observation suggests that, in certain scenarios, prioritizing high-probability transcriptions through the policy gradient is especially beneficial. While EM-tok-b proves to be a robust method across diverse domains, the complex interaction between decoding strategies and adaptation objectives remains an important direction for future study.

\subsection{Different Adaptation Steps}
We analyze how the number of adaptation steps affects performance. We evaluate 1, 3, 5, 10, and 20 adaptation steps on the LS-GS-10 dataset. Figure~\ref{fig:step-ablate} compares the proposed TTA methods EM-seq, EM-tok, and EM-tok-b with Greedy-EM.

\begin{figure}
  \centering
  \includegraphics[width=\linewidth]{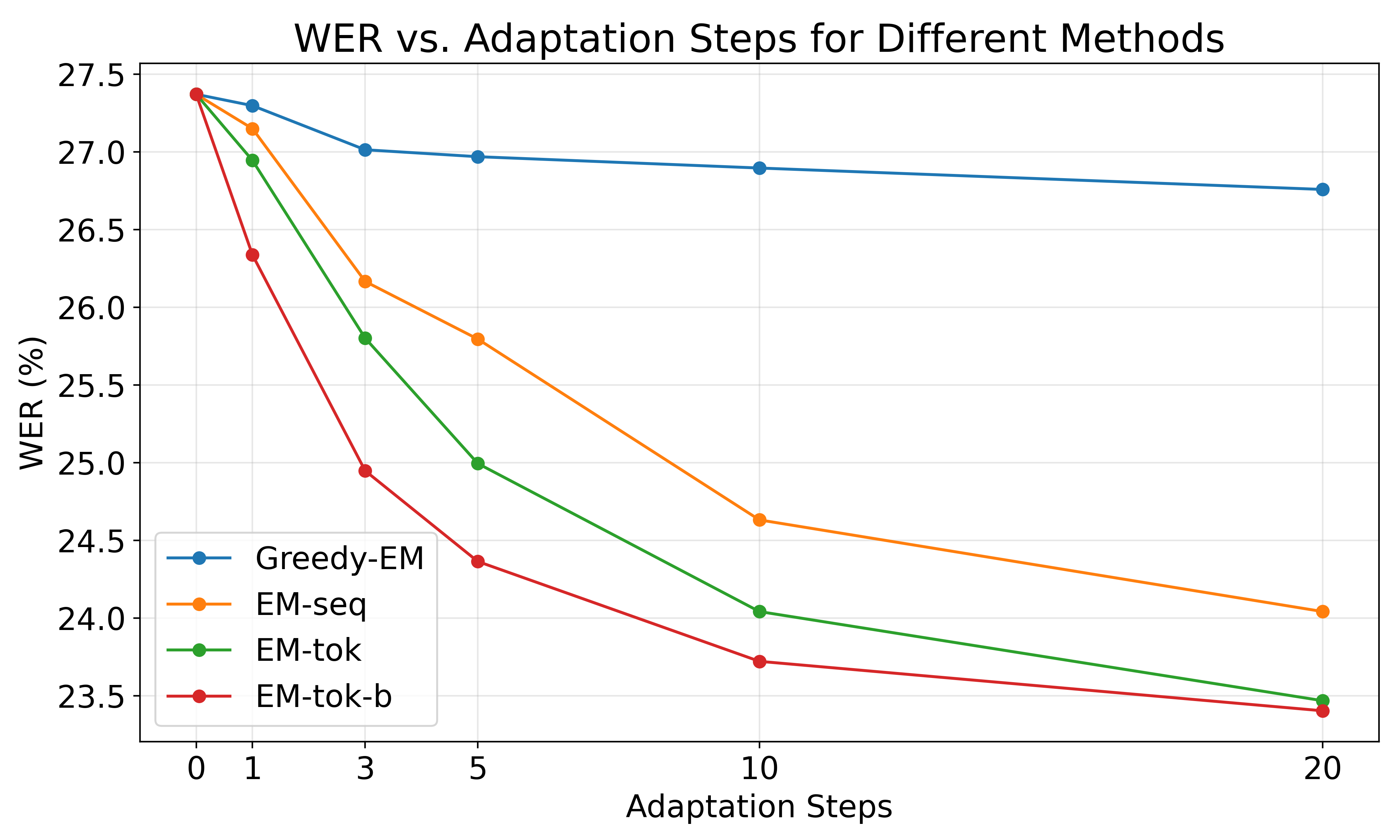}
  \caption{WER(\%) of different TTA methods over different adaptation steps on LS-GS-10.}
  \label{fig:step-ablate}
\end{figure}

Overall, the results show that increasing the number of adaptation steps consistently reduces WER across all methods. Among the proposed methods, EM-tok-b achieves the best performance, followed by EM-tok and EM-seq, demonstrating the advantage of token-level optimization and the additional benefit brought by using beam search transcriptions. Moreover, EM-tok-b consistently surpasses all other methods even with fewer adaptation steps, showing strong sample efficiency. The results indicate that incorporating beam search transcriptions provides higher-quality optimization signals, which substantially enhance overall effectiveness.

\subsection{Qualitative Analysis}
To provide an intuitive understanding of these improvements, Table \ref{tab:case} presents representative examples where the EM-tok-b method corrects errors made by the source model. We observe that EM-tok-b primarily resolves three types of errors. First, it fixes grammatical mistakes caused by incorrect word boundary recognition (e.g., misinterpreting ``athlete" as ``a fleet"). Second, it restores the intended meaning from phonetically similar but meaningless transcriptions. Third, it breaks persistent repetition loops during decoding.
\input{tables/case}

\subsection{Efficiency Analysis}
\label{sec:eff}
Figure~\ref{fig:time} compares the WER against the average runtime (in seconds) for a 1-second utterance across different methods on the LS-GS-10 dataset. We evaluate the proposed TTA methods with sample sizes $G \in \{4, 16, 64\}$ and Greedy-EM.

Our results demonstrate that EM-tok-b with $G=16$ provides a good balance between adaptation performance and computational overhead. While Greedy-EM is computationally cheaper, our proposed methods achieve significantly lower WER. We observe that EM-tok-b is notably more efficient than EM-tok under the same sample size. This efficiency gain stems from the fact that beam search typically yields more stable and concise transcriptions. In contrast, random sampling can produce significantly longer transcriptions, increasing the maximum sequence length within a batch and leading to higher computational costs during the process. For instance, at $G=16$, EM-tok-b reduces the total runtime by nearly 35\% compared to EM-tok while maintaining competitive performance. These findings suggest that utilizing beam search transcriptions is a highly practical trick.

\begin{figure}
  \centering
  \includegraphics[width=0.95\linewidth]{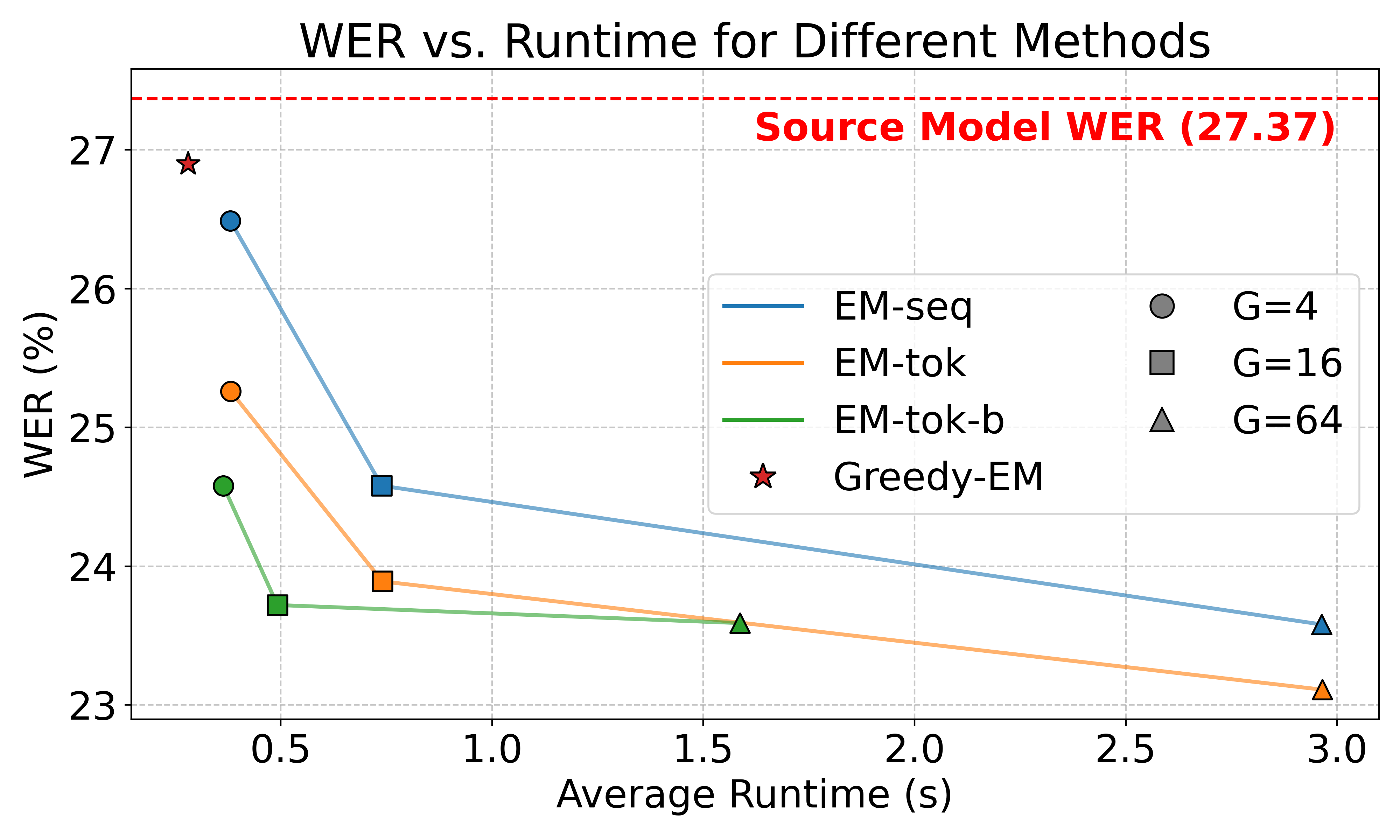}
    \caption{Comparison of WER and average runtime for a 1-second utterance for different methods on LS-GS-10. The dashed red line marks the WER of the unadapted source model. 
    }
    \label{fig:time}
\end{figure}

\subsection{Different Model Sizes}
We validate the proposed method on different sizes of Whisper models. We use the LS-GS-10 dataset. Table~\ref{tab:model-ablate} summarizes the results. Results show that our proposed methods consistently outperform the baselines across different model sizes, yielding significantly better WER than the source model itself.

\input{tables/model-ablate}

\section{Conclusion}
In this work, we resolve the theoretical ambiguity surrounding EM in TTA for autoregressive models by deriving a mathematically correct gradient expression. We introduce the complete EM objective and unify previously fragmented heuristics within a principled framework. Through extensive evaluation on TTA for Whisper ASR across more than 20 diverse domains, we demonstrate that our theoretically grounded formulation is effective across diverse settings and offer important insights into the behavior of the proposed objective. Overall, our findings establish a solid theoretical foundation for EM in TTA for autoregressive models.

\section{Acknowledgement}   
This work was supported by the Ministry of Education (MOE) of Taiwan under the project Taiwan Centers of Excellence in Artificial Intelligence, through the NTU Artificial Intelligence Center of Research Excellence (NTU AI-CoRE)”.
 
We thank the National Center for High-performance Computing (NCHC) of the National Applied Research Laboratories (NARLabs) in Taiwan for providing computational and storage resources.

\section{Generative AI Use Disclosure}
Generative AI tools assisted in the linguistic polishing of the manuscript. 
The authors remain solely responsible for the research design, experiments, analysis, and reported results. 
AI tools did not contribute to the substantive scientific content.

\bibliographystyle{IEEEtran}
\bibliography{mybib}

\clearpage
\section{Appendix}
\subsection{The Special Case of Non-Autoregressive ASR}
Our framework also provides a theoretical perspective on earlier TTA methods for non-autoregressive ASR~\cite{suta}. In non-autoregressive models, the output at each frame is independent of other frames. In this setting, the token-level entropy estimator, $\hat{H}_{tok}(\bm{y})$, does not depend on a particular output sequence $\bm{y}$ and therefore becomes constant across all possible transcriptions.
\begin{subequations}
    \begin{align}
        &\mathbb{E}_{\bm{y}\sim \pi_{\theta}}[\hat{H}_{tok}(\bm{y})] \\
        &= \mathbb{E}_{\bm{y}\sim \pi_{\theta}}\left[-\sum_{t=1}^{|\bm{y}|}\sum_{j\in\mathcal{V}}\pi_{\theta}(j|\bm{y}_{<t})\log\pi_{\theta}(j|\bm{y}_{<t})\right] \\
        &= \mathbb{E}_{\bm{y}\sim \pi_{\theta}}\left[-\sum_{t=1}^{|\bm{y}|}\sum_{j\in\mathcal{V}}\pi_{\theta}(y_t=j)\log\pi_{\theta}(y_t=j)\right] \\
        &= -\sum_{t=1}^{|\bm{y}|}\sum_{j\in\mathcal{V}}\pi_{\theta}(y_t=j)\log\pi_{\theta}(y_t=j)
        \label{eq:NARASR}
    \end{align}
\end{subequations}

Consequently, the expectation term vanishes, and EM reduces to directly minimizing the token-level entropy estimator. This confirms that, if we ignore the \textit{blank-symbol collapse} in non-autoregressive ASR (i.e., treating $c\phi\phi a\phi t$ and $ca\phi\phi t\phi$ as distinct transcriptions), directly adopting Eq.~\ref{eq:NARASR} as an objective as in \cite{suta} is theoretically sound.

With our framework, it is now straightforward to deduce the correct EM for non-autoregressive ASR, given the collapse rule. Using the sequence-level estimator and the results in Section~\ref{sec:math-em-seq}, we show that EM is equivalent to reinforcement learning with cost $-\log \pi_{\theta}(\bm{y})$. Notably, $-\log \pi_{\theta}(\bm{y})$ corresponds exactly to the CTC objective. This leads to a key insight: EM is equivalent to \textbf{performing reinforcement learning on sampled outputs with the CTC loss as the cost}.

\section{Limitation and Future Work}
Despite the theoretical and empirical advantages of our unified EM objective, several limitations remain. First, the requirement for multiple samples to ensure accurate gradient estimation increases the computational overhead during the adaptation phase. Compared to the teacher-forcing heuristic, which only requires a single sample, our approach is much slower and requires more VRAM, which is impractical for real-time deployment. Second, while our findings are mathematically general, our experimental validation in this work focuses exclusively on TTA for the Whisper ASR. The performance characteristics of the derived gradient decomposition may vary across different architectures and tasks.

The potential for future research following this work is extensive. A straightforward direction involves exploring the non-trivial interaction between decoding strategies and adaptation objectives, particularly understanding why specific domains favor certain methods. Additionally, developing techniques to accelerate the sampling process or utilizing more efficient gradient estimators could reduce the current computational costs.

Finally, since we have identified that many existing adaptation methods utilize incomplete gradient expressions, there is a significant opportunity to apply our approach to fix the current methods to improve their performance.

\end{document}

%% file: tables/link.tex
\begin{table*}[t]
  \caption{Comparison with EM objectives and corresponding gradients in previous works. Note that while $\mathcal{L}_{EM}^{seq}$ and $\mathcal{L}_{EM}^{tok}$ appear as distinct formulations, they are both correct since their respective gradients with respect to $\theta$ are both unbiased estimators for EM.}
  \label{tab:link}
  \centering
  % Increase vertical padding (1.5 is usually good for fractions/gradients)
  \renewcommand{\arraystretch}{1.3}
  \begin{tabular}{ccccc}
    \toprule
    \textbf{Work} & \textbf{Method Name} & \textbf{Gradient w.r.t. $\theta$} & \textbf{Corresponding Objective (ours)} & \textbf{Correctness} \\
    \midrule
    \cite{gao2025oneshotentropyminimization, sgem, cea, slmtta} & - & $\sum_{t=1}^{|\bm{y}|}\nabla_{\theta}\mathcal{H}(\pi_{\theta}(\cdot|\bm{y}_{<t}))$ & $\mathcal{L}_{ENT}^{tok}$ & \textcolor{red}{\ding{55}} \\ 
    \multirow{3}{*}{\cite{agarwal2025unreasonable}} & EM-FT & $\sum_{t=1}^{|\bm{y}|}\nabla_{\theta}\mathcal{H}(\pi_{\theta}(\cdot|\bm{y}_{<t}))$ & $\mathcal{L}_{ENT}^{tok}$ & \textcolor{red}{\ding{55}} \\
     & EM-RL-token & $\sum_{t=1}^{|\bm{y}|}\mathcal{H}(\pi_{\theta}(\cdot|\bm{y}_{<t}))\nabla_{\theta}\log\pi_{\theta}(\bm{y})$ & $\mathcal{L}_{PG}^{tok}$ & \textcolor{red}{\ding{55}} \\
     & EM-RL-sequence & $-\sum_{t=1}^{|\bm{y}|}\log\pi_{\theta}(y_t|\bm{y}_{<t})\nabla_{\theta}\log\pi_{\theta}(\bm{y})$ & $\mathcal{L}_{EM}^{seq}=\mathcal{L}_{PG}^{seq}$ & \textcolor{green}{\ding{51}} \\
     \midrule
     \multirow{2}{*}{Ours} & \multirow{2}{*}{-} &
        $\sum_{t=1}^{|\bm{y}|}\mathcal{H}(\pi_{\theta}(\cdot|\bm{y}_{<t}))\nabla_{\theta}\log\pi_{\theta}(\bm{y})$
     & \multirow{2}{*}{$\mathcal{L}_{EM}^{tok}=\mathcal{L}_{PG}^{tok}+\mathcal{L}_{ENT}^{tok}$} & \multirow{2}{*}{\textcolor{green}{\ding{51}}} \\
     & & $+ \sum_{t=1}^{|\bm{y}|}\nabla_{\theta}\mathcal{H}(\pi_{\theta}(\cdot|\bm{y}_{<t}))$ & & \\
    \bottomrule
  \end{tabular}
\end{table*}

%% file: tables/ls-big-table.tex
% tables/ls-big-table.tex
\begingroup
\small
\begin{tabular}{lccccccccccc}
    \toprule
    \textbf{Method} & \textbf{LS-AA-10} & \textbf{LS-AC-10} & \textbf{LS-BA-10} & \textbf{LS-CM-10} & \textbf{LS-GS-10} & \textbf{LS-MU-10} & \textbf{LS-NB-10} & \textbf{LS-SD-10} & \textbf{LS-TP-10} & \textbf{LS-VC-10} & \textbf{Avg.} \\
    \midrule
    \textbf{Source}      & 17.63 & 15.23 & 22.79 & 24.25 & 27.37 & 28.36 & 35.32 & 15.17 & 17.77 & 21.36 & 22.53 \\
    \textbf{Greedy-EM}   & 17.30 & 15.08 & 22.29 & 23.59 & 26.41 & 27.46 & 33.75 & 14.82 & 17.09 & 21.32 & 21.91 \\
    \midrule
    \textbf{EM-seq}      & 16.50 & 14.46 & 21.46 & 23.48 & 25.81 & 26.68 & 34.20 & 14.14 & 16.50 & 20.16 & 21.34 \\
    \textbf{EM-tok}      & 16.05 & 14.08 & 20.96 & 22.56 & 24.95 & 26.06 & 33.55 & 13.69 & 16.10 & 19.65 & 20.77 \\
    \textbf{EM-tok-b}    & \textbf{14.91} & \textbf{12.80} & \textbf{19.31} & \textbf{20.95} & \textbf{23.16} & \textbf{24.03} & \textbf{30.81} & \textbf{12.87} & \textbf{14.86} & \textbf{17.75} & \textbf{19.15} \\
    \bottomrule
\end{tabular}
\endgroup

%% file: tables/l2-big-table.tex
% tables/l2-table.tex
\begingroup
\small
\begin{tabular}{lccccccc}
    \toprule
    \textbf{Method} & \textbf{Arabic} & \textbf{Chinese} & \textbf{Hindi} & \textbf{Korean} & \textbf{Spanish} & \textbf{Vietnamese} & \textbf{Avg.} \\
    \midrule
    \textbf{Source}      & 19.22 & 22.71 & 10.39 & 14.56 & 18.31 & 30.89 & 19.35 \\
    \textbf{Greedy-EM}   & 18.45 & 19.67 & 10.27 & 14.23 & 19.47 & 30.53 & 18.77 \\
    \midrule
    \textbf{EM-seq}      & 16.82 & 20.23 &  9.52 & 12.81 & 18.04 & 28.67 & 17.68 \\
    \textbf{EM-tok}      & 16.21 & \textbf{18.83} &  9.30 & 12.45 & 17.57 & 27.95 & 17.05 \\
    \textbf{EM-tok-b}    & \textbf{15.54} & 18.84 & \textbf{8.63} & \textbf{11.86} & \textbf{15.19} & \textbf{27.17} & \textbf{16.21} \\
    \bottomrule
\end{tabular}
\endgroup

%% file: tables/mls-big-table.tex
% tables/mls-table.tex
\begingroup
\small
\begin{tabular}{lcccccccc}
    \toprule
    \textbf{Method} & \textbf{Dutch} & \textbf{French} & \textbf{German} & \textbf{Italian} & \textbf{Polish} & \textbf{Portuguese} & \textbf{Spanish} & \textbf{Avg.} \\
    \midrule
    \textbf{Source}      & 30.88 & 24.75 & 19.85 & 32.86 & 25.31 & 23.98 & 14.42 & 24.58 \\
    \textbf{Greedy-EM}   & 30.65 & 24.53 & 19.14 & 32.27 & 24.87 & 23.13 & 14.00 & 24.08 \\
    \midrule
    \textbf{EM-seq}      & 30.67 & 24.17 & 19.55 & 32.04 & 25.12 & 23.62 & 14.05 & 24.17 \\
    \textbf{EM-tok}      & 30.58 & 24.09 & 19.51 & 31.92 & 24.84 & 23.74 & 13.83 & 24.07 \\
    \textbf{EM-tok-b}    & \textbf{29.39} & \textbf{23.12} & \textbf{17.70} & \textbf{30.59} & \textbf{23.20} & \textbf{21.75} & \textbf{12.67} & \textbf{22.63} \\
    \bottomrule
\end{tabular}
\endgroup

%% file: tables/beam-search.tex
\begin{table}
    \centering
    \caption{WER (\%) comparison between the proposed TTA method and beam search. We only report partial results on two randomly selected domains from each dataset.}
    \label{tab:em-vs-beam}
    %\adjustbox{width=1.0\linewidth}{
        \begin{tabular}{cccc}
            \toprule
            \textbf{Dataset} & \textbf{Greedy} & \textbf{Beam Search} & \textbf{EM-tok-b} \\
            \midrule
            LS-GS-10 & 27.37 & 24.72 & \textbf{23.16} \\
            LS-TP-10 & 17.77 & 16.17 & \textbf{14.86} \\
            L2-Korean & 14.56 & 13.5 & \textbf{11.86} \\
            L2-Spanish & 18.31 & 17.03 & \textbf{15.19} \\
            MLS-Italian & 32.86 & 31.83 & \textbf{30.59} \\
            MLS-Polish & 25.31 & \textbf{22.08} & 23.2 \\
            \bottomrule
        \end{tabular}
    %}
\end{table}

%% file: tables/case.tex
\begin{table}[th]
\caption{Qualitative examples illustrating baseline errors and our proposed corrections. \textcolor{red}{Red} indicates errors, and \textcolor{ForestGreen}{green} denotes corrected or successfully recovered tokens.}
\label{tab:case}
\centering
\begin{tabular}{@{}p{0.95\columnwidth}@{}}
\toprule
\textbf{Example and Description} \\
\midrule
\textbf{Ground Truth:} he was \textcolor{ForestGreen}{an athlete} and \textcolor{ForestGreen}{a giant} \\
\textbf{Source:} he was \textcolor{red}{in a fleet} and \textcolor{red}{agent} \\
\textbf{EM-tok-b:} he was \textcolor{ForestGreen}{an athlete} and \textcolor{ForestGreen}{a giant} \\
\textit{Description: Corrects improper splitting and merging of short words in continuous speech.} \\
\midrule
\textbf{Ground Truth:} she was \textcolor{ForestGreen}{attractive and impertinent} especially the \textcolor{ForestGreen}{latter} \\
\textbf{Source:} she was \textcolor{red}{a truck dude i made a pair of doing it} especially the \textcolor{red}{matter} \\
\textbf{EM-tok-b:} she was \textcolor{ForestGreen}{attractive and impertinent} especially the \textcolor{ForestGreen}{latter} \\
\textit{Description: Restores the intended meaning from phonetically similar but meaningless transcriptions.} \\
\midrule
\textbf{Ground Truth:} \textcolor{ForestGreen}{a quarter past 10 half past} \\
\textbf{Source:} \textcolor{red}{one one 2 3 4 5 5 5 5 5 5 5 5 \dots} \\
\textbf{EM-tok-b:} \textcolor{ForestGreen}{a quarter past 10 half past} \\
\textit{Description: Halts continuous repetition errors.} \\
\bottomrule
\end{tabular}
\end{table}

%% file: tables/model-ablate.tex
\begin{table}[htbp]
    \centering
    \caption{WER (\%) comparison of different TTA methods on different Whisper model sizes. We use the LS-GS-10 dataset.}
    \label{tab:model-ablate}
    %\adjustbox{width=1.0\linewidth}{
        \begin{tabular}{lccc}
            \toprule
            \textbf{Method} & \textbf{Tiny} & \textbf{Small} & \textbf{Large} \\
            \midrule
            \textbf{Source} & 37.4 & 16.47 & 8.54 \\
            \textbf{Greedy-EM} & 36.59 & 15.94 & 8.21 \\
            \midrule
            \textbf{EM-seq} & 33.5 & 15.58 & 8.5 \\
            \textbf{EM-tok} & \textbf{33.1} & 15.11 & 7.86 \\
            \textbf{EM-tok-b} & 33.21 & \textbf{14.57} & \textbf{7.56} \\
            \bottomrule
        \end{tabular}
    %}
\end{table}

%% file: mybib.bib
@article{REINFORCE,
author = {Williams, Ronald J.},
title = {Simple Statistical Gradient-Following Algorithms for Connectionist Reinforcement Learning},
year = {1992},
issue_date = {May 1992},
publisher = {Kluwer Academic Publishers},
address = {USA},
volume = {8},
number = {3–4},
issn = {0885-6125},
url = {https://doi.org/10.1007/BF00992696},
doi = {10.1007/BF00992696},
journal = {Mach. Learn.},
month = may,
pages = {229–256},
numpages = {28}
}

@inproceedings{NEURIPS2020_92d1e1eb,
 author = {Baevski, Alexei and Zhou, Yuhao and Mohamed, Abdelrahman and Auli, Michael},
 booktitle = {Advances in Neural Information Processing Systems},
 editor = {H. Larochelle and M. Ranzato and R. Hadsell and M.F. Balcan and H. Lin},
 pages = {12449--12460},
 publisher = {Curran Associates, Inc.},
 title = {wav2vec 2.0: A Framework for Self-Supervised Learning of Speech Representations},
 url = {https://proceedings.neurips.cc/paper_files/paper/2020/file/92d1e1eb1cd6f9fba3227870bb6d7f07-Paper.pdf},
 volume = {33},
 year = {2020}
}

@misc{burchi2021efficientconformerprogressivedownsampling,
      title={Efficient conformer: Progressive downsampling and grouped attention for automatic speech recognition}, 
      author={Maxime Burchi and Valentin Vielzeuf},
      year={2021},
      eprint={2109.01163},
      archivePrefix={arXiv},
      primaryClass={eess.AS},
      url={https://arxiv.org/abs/2109.01163}, 
}

@inproceedings{memo,
author = {Zhang, Marvin and Levine, Sergey and Finn, Chelsea},
title = {MEMO: test time robustness via adaptation and augmentation},
year = {2022},
isbn = {9781713871088},
publisher = {Curran Associates Inc.},
address = {Red Hook, NY, USA},
articleno = {2799},
numpages = {14},
location = {New Orleans, LA, USA},
series = {NIPS '22}
}

@inproceedings{tent,
  title={Tent: Fully Test-Time Adaptation by Entropy Minimization},
  author={Wang, Dequan and Shelhamer, Evan and Liu, Shaoteng and Olshausen, Bruno and Darrell, Trevor},
  booktitle={International Conference on Learning Representations},
  year={2020}
}

@inproceedings{eata,
  title={Efficient test-time model adaptation without forgetting},
  author={Niu, Shuaicheng and Wu, Jiaxiang and Zhang, Yifan and Chen, Yaofo and Zheng, Shijian and Zhao, Peilin and Tan, Mingkui},
  booktitle={International conference on machine learning},
  pages={16888--16905},
  year={2022},
  organization={PMLR}
}

@inproceedings{cotta,
  title={Continual test-time domain adaptation},
  author={Wang, Qin and Fink, Olga and Van Gool, Luc and Dai, Dengxin},
  booktitle={Proceedings of the IEEE/CVF Conference on Computer Vision and Pattern Recognition},
  pages={7201--7211},
  year={2022}
}

@inproceedings{
sar,
title={Towards Stable Test-time Adaptation in Dynamic Wild World},
author={Shuaicheng Niu and Jiaxiang Wu and Yifan Zhang and Zhiquan Wen and Yaofo Chen and Peilin Zhao and Mingkui Tan},
booktitle={The Eleventh International Conference on Learning Representations },
year={2023},
url={https://openreview.net/forum?id=g2YraF75Tj}
}

@inproceedings{suta,
  author={Guan-Ting Lin and Shang-Wen Li and Hung-yi Lee},
  title={{Listen, Adapt, Better WER: Source-free Single-utterance Test-time Adaptation for Automatic Speech Recognition}},
  year=2022,
  booktitle={Proc. Interspeech 2022},
  pages={2198--2202},
  doi={10.21437/Interspeech.2022-600},
  issn={2308-457X}
}

@inproceedings{sgem,
  author={Changhun Kim and Joonhyung Park and Hajin Shim and Eunho Yang},
  title={{SGEM: Test-Time Adaptation for Automatic Speech Recognition via Sequential-Level Generalized Entropy Minimization}},
  year=2023,
  booktitle={Proc. INTERSPEECH 2023},
  pages={3367--3371},
  doi={10.21437/Interspeech.2023-1282},
  issn={2308-457X}
}

@inproceedings{litta,
  title     = {LI-TTA: Language Informed Test-Time Adaptation for Automatic Speech Recognition},
  author    = {Eunseop Yoon and Hee Suk Yoon and John Harvill and Mark Hasegawa-Johnson and Chang D. Yoo},
  year      = {2024},
  booktitle = {Interspeech 2024},
  pages     = {3490--3494},
  doi       = {10.21437/Interspeech.2024-1829},
  issn      = {2958-1796},
}

@inproceedings{cea,
    title = "Advancing Test-Time Adaptation in Wild Acoustic Test Settings",
    author = "Liu, Hongfu  and
      Huang, Hengguan  and
      Wang, Ye",
    editor = "Al-Onaizan, Yaser  and
      Bansal, Mohit  and
      Chen, Yun-Nung",
    booktitle = "Proceedings of the 2024 Conference on Empirical Methods in Natural Language Processing",
    month = nov,
    year = "2024",
    address = "Miami, Florida, USA",
    publisher = "Association for Computational Linguistics",
    url = "https://aclanthology.org/2024.emnlp-main.405/",
    doi = "10.18653/v1/2024.emnlp-main.405",
    pages = "7138--7155",
}

@inproceedings{dsuta,
    title = "Continual Test-time Adaptation for End-to-end Speech Recognition on Noisy Speech",
    author = "Lin, Guan-Ting  and
      Huang, Wei Ping  and
      Lee, Hung-yi",
    editor = "Al-Onaizan, Yaser  and
      Bansal, Mohit  and
      Chen, Yun-Nung",
    booktitle = "Proceedings of the 2024 Conference on Empirical Methods in Natural Language Processing",
    month = nov,
    year = "2024",
    address = "Miami, Florida, USA",
    publisher = "Association for Computational Linguistics",
    url = "https://aclanthology.org/2024.emnlp-main.1116/",
    doi = "10.18653/v1/2024.emnlp-main.1116",
    pages = "20003--20015",
}

@misc{suta-lm,
      title={SUTA-LM: Bridging Test-Time Adaptation and Language Model Rescoring for Robust ASR}, 
      author={Wei-Ping Huang and Guan-Ting Lin and Hung-yi Lee},
      year={2025},
      eprint={2506.11121},
      archivePrefix={arXiv},
      primaryClass={cs.CL},
      url={https://arxiv.org/abs/2506.11121}, 
}

@inproceedings{rloo,
    title = "Back to Basics: Revisiting {REINFORCE}-Style Optimization for Learning from Human Feedback in {LLM}s",
    author = {Ahmadian, Arash  and
      Cremer, Chris  and
      Gall{\'e}, Matthias  and
      Fadaee, Marzieh  and
      Kreutzer, Julia  and
      Pietquin, Olivier  and
      {\"U}st{\"u}n, Ahmet  and
      Hooker, Sara},
    editor = "Ku, Lun-Wei  and
      Martins, Andre  and
      Srikumar, Vivek",
    booktitle = "Proceedings of the 62nd Annual Meeting of the Association for Computational Linguistics (Volume 1: Long Papers)",
    month = aug,
    year = "2024",
    address = "Bangkok, Thailand",
    publisher = "Association for Computational Linguistics",
    url = "https://aclanthology.org/2024.acl-long.662/",
    doi = "10.18653/v1/2024.acl-long.662",
    pages = "12248--12267"
}

@inproceedings{
dapo,
title={{DAPO} : Improving Multi-Step Reasoning Abilities of Large Language Models with Direct Advantage-Based Policy Optimization},
author={Jiacai Liu and Chaojie Wang and Chris Yuhao Liu and Liang Zeng and Rui Yan and Yiwen Sun and Yang Liu},
booktitle={The Thirty-ninth Annual Conference on Neural Information Processing Systems},
year={2025},
url={https://openreview.net/forum?id=77eEDRhPkQ}
}

@misc{whisper,
      title={Robust Speech Recognition via Large-Scale Weak Supervision}, 
      author={Alec Radford and Jong Wook Kim and Tao Xu and Greg Brockman and Christine McLeavey and Ilya Sutskever},
      year={2022},
      eprint={2212.04356},
      archivePrefix={arXiv},
      primaryClass={eess.AS},
      url={https://arxiv.org/abs/2212.04356}, 
}

@misc{gao2025oneshotentropyminimization,
      title={One-shot Entropy Minimization}, 
      author={Zitian Gao and Lynx Chen and Haoming Luo and Joey Zhou and Bryan Dai},
      year={2025},
      eprint={2505.20282},
      archivePrefix={arXiv},
      primaryClass={cs.CL},
      url={https://arxiv.org/abs/2505.20282}, 
}

@article{agarwal2025unreasonable,
  title={The unreasonable effectiveness of entropy minimization in llm reasoning},
  author={Agarwal, Shivam and Zhang, Zimin and Yuan, Lifan and Han, Jiawei and Peng, Hao},
  journal={arXiv preprint arXiv:2505.15134},
  year={2025}
}

@article{slot,
  publtype={informal},
  author={Yang Hu and Xingyu Zhang and Xueji Fang and Zhiyang Chen and Xiao Wang and Huatian Zhang and Guojun Qi},
  title={SLOT: Sample-specific Language Model Optimization at Test-time},
  year={2025},
  month={May},
  cdate={1746057600000},
  journal={CoRR},
  volume={abs/2505.12392},
  url={https://doi.org/10.48550/arXiv.2505.12392}
}

@inproceedings{
kuhn2023semantic,
title={Semantic Uncertainty: Linguistic Invariances for Uncertainty Estimation in Natural Language Generation},
author={Lorenz Kuhn and Yarin Gal and Sebastian Farquhar},
booktitle={The Eleventh International Conference on Learning Representations },
year={2023},
url={https://openreview.net/forum?id=VD-AYtP0dve}
}

@article{empo,
  title={Right Question is Already Half the Answer: Fully Unsupervised LLM Reasoning Incentivization},
  author={Zhang, Qingyang and Wu, Haitao and Zhang, Changqing and Zhao, Peilin and Bian, Yatao},
  journal={Advances in neural information processing systems},
  year={2025}
}

@article{seedgrpo,
  title={Seed-grpo: Semantic entropy enhanced grpo for uncertainty-aware policy optimization},
  author={Chen, Minghan and Chen, Guikun and Wang, Wenguan and Yang, Yi},
  journal={arXiv preprint arXiv:2505.12346},
  year={2025}
}

@misc{slmtta,
      title={SLM-TTA: A Framework for Test-Time Adaptation of Generative Spoken Language Models}, 
      author={Yuan-Kuei Wu and Yang Liu and Yiteng Huang and Zhaojun Yang and Haibin Wu and Ruizhe Huang and Yi-Te and Hsu and Shuyu Kong and Ming Sun and Florian Metze and Li Wan},
      year={2025},
      eprint={2512.24739},
      archivePrefix={arXiv},
      primaryClass={cs.SD},
      url={https://arxiv.org/abs/2512.24739}, 
}

@inproceedings{l2arctic,
  title     = {L2-ARCTIC: A Non-native English Speech Corpus},
  author    = {Guanlong Zhao and Sinem Sonsaat and Alif Silpachai and Ivana Lucic and Evgeny Chukharev-Hudilainen and John Levis and Ricardo Gutierrez-Osuna},
  year      = {2018},
  booktitle = {Interspeech 2018},
  pages     = {2783--2787},
  doi       = {10.21437/Interspeech.2018-1110},
  issn      = {2958-1796},
}

@article{mls,
  title={MLS: A Large-Scale Multilingual Dataset for Speech Research},
  author={Vineel Pratap and Qiantong Xu and Anuroop Sriram and Gabriel Synnaeve and Ronan Collobert},
  journal={ArXiv},
  year={2020},
  volume={abs/2012.03411}
}

@inproceedings{cv,
  author = {Ardila, R. and Branson, M. and Davis, K. and Henretty, M. and Kohler, M. and Meyer, J. and Morais, R. and Saunders, L. and Tyers, F. M. and Weber, G.},
  title = {Common Voice: A Massively-Multilingual Speech Corpus},
  booktitle = {Proceedings of the 12th Conference on Language Resources and Evaluation (LREC 2020)},
  pages = {4211--4215},
  year = 2020
}
